\documentclass[a4paper]{llncs}
\usepackage{latexsym,amssymb,upgreek}
\usepackage{xspace}
\usepackage{enumerate}
\usepackage{mathrsfs}

\long\def\COMMENT#1{}

\newcommand{\And}{\wedge}
\newcommand{\Or}{\vee}
\newcommand{\compl}[1]{\overline{#1}}

\newcommand{\false}{\mathbf{false}}
\newcommand{\true}{\mathbf{true}}

\newcommand{\comment}[1]{}

\def\eod {{\unskip\nobreak\hfil\penalty50
\hskip2em\hbox{}\nobreak\hfil $\Box$
\parfillskip=0pt \finalhyphendemerits=0 \par \medskip}}

\def\qed {{\unskip\nobreak\hfil\penalty50
\hskip2em\hbox{}\nobreak\hfil \rule{2mm}{2mm}
\parfillskip=0pt \finalhyphendemerits=0 \par \medskip}}

\newcommand{\model}{\ensuremath{\mbox{\boldmath $\mathcal{M}$}}\xspace}

\newcommand{\pow}{\mbox{\rm pow}}
\newcommand{\Rel}{\mbox{\rm Rel}}

\newcommand{\SubF}{\mathit{SubF}}
\newcommand{\Kr}{\mathit{K}}
\newcommand{\K}{\mathsf{K}}
\newcommand{\T}{\mathsf{T}}
\newcommand{\Ac}{\mathbf{5}}
\newcommand{\B}{\mathsf{B}}
\newcommand{\Aq}{\mathbf{4}}
\newcommand{\D}{\mathsf{D}}
\newcommand{\Sc}{\mathsf{S5}}

\newcommand{\TSc}{\tau_{\Sc}}
\newcommand{\TScuno}{\tau_{\Sc}^{1}}
\newcommand{\TScdue}{\tau_{\Sc}^{2}}

\newcommand{\TLQSR}{\ensuremath{\mbox{$3\mathit{LQS}^{R}$}}\xspace}
\newcommand{\TLQS}{\ensuremath{\mbox{$3\mathit{LQS}$}}\xspace}
\newcommand{\MLS}{\ensuremath{\mbox{$\mathit{MLS}$}}\xspace}

\newcommand{\Hf}{H_{\varphi}}

\begin{document}

\title{On the satisfiability problem for a 3-level quantified
syllogistic}

\author{Domenico Cantone \and Marianna Nicolosi Asmundo}

\institute{ %
{Universit\`a di Catania, Dipartimento di Matematica e Informatica\\
~email:~\texttt{cantone@dmi.unict.it, nicolosi@dmi.unict.it}}}

\date{}

\maketitle

\begin{abstract}
We show that a collection of three-sorted set-theoretic formulae,
denoted $\TLQSR$ and which admits a restricted form of quantification
over individual and set variables, has a solvable satisfiability
problem by proving that it enjoys a small model property, i.e., any
satisfiable $\TLQSR$-formula $\psi$ has a finite model whose size
depends solely on the size of $\psi$ itself.
We also introduce the sublanguages $(\TLQSR)^{h}$ of $\TLQSR$, whose formulae are characterized by
having quantifier prefixes of length bounded by $h \geq 2$ and some other syntactic constraints,
and we prove that each of them has the satisfiability problem \textsf{NP}-complete.
Then, we show
that the modal logic $\Sc$ can be formalized in $(\TLQSR)^{3}$.
\end{abstract}


%



\section{Introduction}
\emph{Computable set theory} is a research field active since the late
seventies.  Its initial goal was the design of effective
decision procedures to be implemented in theorem provers/verifiers,
for larger and larger collections of set-theoretic formulae (also
called \emph{syllogistics}).  During the years, however, due to the
production of several decidability results of a purely theoretical
nature the main emphasis shifted to the foundational goal of
narrowing the boundary between the decidable and the undecidable in
set theory.

The main results in computable set theory up to 2001 have been
collected in \cite{CFO89,COO01}.  We also mention that the most
efficient decision procedures have been implemented in the proof
verifier \textsf{{\AE}tnaNova} \cite{CPS06,SOC07} and within one of
versions of the system \textsf{STeP} \cite{STeP}.

The basic set-theoretic fragment is the so-called Multi-Level
Syllogistic ($\MLS$, for short) which involves in addition to variables
varying over the von Neumann universe of sets and to propositional
connectives also the basic set-theoretic operators such as $\cup$,
$\cap$, $\setminus$, and the predicates $=$, $\in$, and $\subseteq$.
MLS was proved decidable in \cite{FOS80} and extended over the years
in several ways by the introduction of various operators, predicates,
and restricted forms of quantification.

Most of the decidability results in computable set theory deal with
one-sorted multi-level syllogistics, namely collections of formulae
involving variables of one type only, ranging over the von Neumann
universe of sets.  On the other hand, few decidability results have
been found for multi-sorted stratified syllogistics, where variables
of several types are allowed.  This, despite of the fact that in many
fields of computer science and mathematics often one deals with
multi-sorted languages.

An efficient decision procedure for the satisfiability of the
Two-Level Syllogistic language ($2LS$), a version of $\MLS$ with
variables of two sorts for individuals and sets of individuals, has
been presented in \cite{FerOm1978}.  Subsequently, in \cite{CanCut90},
the extension of $2LS$ with the singleton operator and the Cartesian
product operator has been proved decidable.  The result has been
obtained by embedding $2LS$ in the class of purely universal formulae
of the elementary theory of relations.  Tarski's and Presburger's
arithmetics extended with sets have been studied in \cite{CCS90}.  The
three-sorted language $3LSSPU$ (Three-Level Syllogistic with
Singleton, with Powerset and general Union) has been proved decidable
in \cite{CanCut93}.  More specifically, $3LSSPU$ has three types of
variables, ranging over individuals, sets of individuals, and
collections of sets of individuals, respectively, and involves the
singleton, powerset, and general union operators, in addition to the
operators and predicates present in $2LS$.

In this paper we present a decidability result for the
satisfiability problem of the set-theoretic language $\TLQSR$
(Three-Level
Quantified Syllogistic with Restricted quantifiers), which is a three-sorted quantified syllogistic 
involving \emph{individual variables}, varying over the elements of
a given nonempty universe $D$, \emph{set variables}, ranging over
subsets of $D$, and \emph{collection variables}, varying over
collections of subsets of $D$.

The language of $\TLQSR$ admits a restricted form of quantification
over individual and set variables.  Its vocabulary contains only the
predicate symbols $=$ and $\in$.  In spite of that, $\TLQSR$ allows
to express several constructs of set theory. Among them, the most
comprehensive one is the set former, which in turn allows to express
other operators like the powerset operator, the singleton operator,
and so on.

We will prove that $\TLQSR$ enjoys a small model property by showing
how one can extract, out of a given model satisfying a
$\TLQSR$-formula $\psi$, another model of $\psi$ of bounded finite
cardinality.  The construction of the finite model is inspired to the
algorithms described in \cite{FerOm1978}, \cite{CanCut90}, and \cite{CanCut93}. 

Then, we introduce the sublanguages $(\TLQSR)^h$,of $\TLQSR$, consisting of $\TLQSR$-formulae
having the quantifier prefixes of size bounded by $h \geq 2$ and satisfying some further syntactic constraints.
 It is
shown that each $(\TLQSR)^h$ has
the satisfiability problem \textsf{NP}-complete and that $(\TLQSR)^3$ can express the normal modal logic $\Sc$.

\smallskip

The paper is organized as follows.  In Section \ref{language}, we
introduce the language $\TLQSR$ and we illustrate its expressiveness.
Subsequently, in Section \ref{machinery} the machinery needed to prove
the decidability result is provided.  In particular, a general
definition of a relativized $\TLQSR$-interpretation is introduced,
together with some useful technical results.  In Section
\ref{satisfiability}, the small model property for $\TLQSR$ is
established, thus solving the satisfiability problem for $\TLQSR$.
Then, in Section \ref{expressiveness}, after some examples illustrating the expressivity
of $\TLQSR$ in set theory, we introduce the sublanguages $(\TLQSR)^h$, show
that they have a \textsf{NP}-complete satisfiability problem, and that $(\TLQSR)^3$ can express the modal logic $\Sc$.
Finally, in Section
\ref{conclusions}, we draw our conclusions.

\section{The language $\TLQSR$}\label{language}
We present the language $\TLQSR$ of our interest as follows.  We begin
by defining in Section \ref{genericlanguage} the syntax and the
semantics of a more general three-level quantified language, denoted
$\TLQS$, which contains $\TLQSR$ as a proper fragment.  Then, in
Section \ref{restrictionquant}, we characterize $\TLQSR$ by means of
suitable restrictions on the usage of quantifiers in formulae of
$\TLQS$.

%
%
\subsection{The more general language $\TLQS$}\label{genericlanguage}

\subsubsection{Syntax of $\TLQS$} The three-level quantified language $\TLQS$ involves\footnote{In the paper, variables often come with numerical subscripts. Other types of subscripts are used in
Section \ref{expressiveness} for variables denoting sets or collections of sets of particular relevance (i.e., $X_{U}$, $A_{\pi,h}$).}
\begin{itemize}
\item[(i)] a collection ${\cal V}_{0}$ of \emph{individual} or \emph{sort 0} variables,
denoted by lower case letters $x,y,z,\ldots$;

\item[(ii)] a collection ${\cal V}_{1}$ of \emph{set} or \emph{sort 1} variables,
denoted by final upper case letters $X,Y,Z,\ldots$;

\item[(iii)] a collection ${\cal V}_{2}$ of \emph{collection} or \emph{sort 2} variables,
denoted by initial upper case letters $A,B,C,\ldots$.

\end{itemize}


The atomic formulae of $\TLQS$ are defined as follows:
\begin{itemize}
\item[(1)] level $0$ atomic formulae:
\begin{itemize}
    \item  $x = y$, for $x,y \in {\cal V}_{0}$;
    \item  $x \in X$, for $x \in {\cal V}_{0}, X \in {\cal V}_{1}$;
\end{itemize}
\item[(2)] level $1$ atomic formulae:
\begin{itemize}
    \item $X = Y$, for $X,Y \in {\cal V}_{1}$;
    \item $X \in A$, for $X \in {\cal V}_{1}, A \in {\cal V}_{2}$;
    \item $(\forall z_1) \ldots (\forall z_n) \varphi_0$, with
    $\varphi_0$ a propositional combination of level $0$ atoms and $z_1,\ldots,z_n$ variables of sort $0$;
\end{itemize}
\item[(3)] level $2$ atomic formulae:
\begin{itemize}
    \item $(\forall Z_1) \ldots (\forall Z_m) \varphi_1$, where
    $\varphi_1$ is a propositional combination of level $0$ and
    level $1$ atoms, and $Z_1,\ldots,Z_m$ are variables of sort 1.
 \end{itemize}
\end{itemize}
Finally, the formulae of $\TLQS$ are all the propositional
combinations of atoms of level $0,1$, and $2$.

%

\subsubsection{Semantics of $\TLQS$} A {\em $\TLQS$-interpretation\/}
is a pair $\model=(D,M)$, where
\begin{itemize}
    \item  $D$ is any nonempty collection of objects, called the
    {\em domain\/} or {\em universe\/} of $\model$, and

    \item  $M$ is an assignment over variables of $\TLQS$ such that

    \begin{itemize}
        \item  $Mx \in D$, for each individual variable $x \in {\cal V}_{0}$;

        \item  $MX \in \pow(D)$, for each set variable $X \in {\cal V}_{1}$;

        \item  $MA \in \pow(\pow(D))$, for all collection variables
        $A \in {\cal V}_{2}$.\footnote{We recall that, for any set
    $s$, $\pow(s)$ denotes the \emph{powerset} of $s$, i.e., the
    collection of all subsets of $s$.}
    \end{itemize}
\end{itemize}

Let
\begin{itemize}
    \item[-]$\model=(D,M)$ be a $\TLQS$-interpretation,
    \item[-]  $x_{1},\ldots,x_{l} \in {\cal V}_{0}$,

    \item[-]  $X_{1},\ldots,X_{m} \in {\cal V}_{1}$,


    \item[-]  $u_{1},\ldots,u_{l} \in D$,

    \item[-]  $U_{1},\ldots,U_{m} \in \pow(D)$,

\end{itemize}

By
\[\model[x_{1}/u_{1},\ldots,x_{l}/u_{l},
         X_{1}/U_{1},\ldots,X_{m}/U_{m},
         ]\,,\]
we denote the interpretation $\model' = (D,M')$ such that
\[
\begin{array}{llll}
    M'x_{i} & = & u_{i}\,, & \mbox{{\rm for} $i = 1,\ldots,l$}  \\
    M'X_{j} & = & U_{j}\,, & \mbox{{\rm for} $j = 1,\ldots,m$}  \\
\end{array}
\]
and which otherwise coincides with $\model$ on all remaining
variables. Throughout the paper we use the abbreviations: ${\cal
M}^{z}$ for $\model[z_{1}/u_{1},\ldots,z_{n}/u_{n}]$ and ${\cal
M}^{Z}$ for $\model[Z_{1}/U_{1},\ldots,Z_{m}/U_{m}]$.

\begin{definition}\label{def:satisf}
Let $\varphi$ be a $\TLQS$-formula and let $\model = (D, M)$ be a
$\TLQS$-interpre\-ta\-tion.  We define the notion of {\em
satisfiability} of $\varphi$ with respect to $\model$ (denoted by
$\model \models \varphi$) inductively as follows
\begin{itemize}
\item[1.] $\model \models x=y$ ~~iff~~ $Mx = My$;
\item[2.] $\model \models x \in X$ ~~iff~~ $Mx \in MX$;
\item[3.] $\model \models X=Y$ ~~iff~~ $MX = MY$;
\item[4.] $\model \models X \in A$ ~~iff~~ $MX \in MA$;
\item[5.] $\model \models (\forall z_1) \ldots (\forall z_n) \varphi_0$
~~iff~~ $\model[z_1/u_1,\ldots , z_n/u_n] \models \varphi_0$, for
all $u_1,\ldots ,u_n \in D$;
\item[6.]  $\model \models (\forall Z_1) \ldots (\forall Z_m)
\varphi_1$ iff $\model[Z_1/U_1,\ldots , Z_m/U_m] \models
\varphi_1$, for all $U_1,\ldots,U_n \in \pow(D)$.

Propositional connectives are interpreted in the standard way,
namely

\item[7.] $\model \models \varphi_{1} \wedge \varphi_{2}$ ~~iff~~
$\model \models \varphi_{1}$ and $\model \models \varphi_{2}$;

\item[8.] $\model \models \varphi_{1} \vee \varphi_{2}$ ~~iff~~
$\model \models \varphi_{1}$ or $\model \models \varphi_{2}$;

\item[9.] $\model \models \neg \varphi$ ~~iff~~
$\model \not \models \varphi$. \eod
\end{itemize}
\end{definition}

Let $\psi$ be a $\TLQS$-formula, if $\model \models \psi$, i.e.
$\model$ {\em satisfies\/} $\psi$, then $\model$ is said to be a
$\TLQS$-{\em model\/} for $\psi$.  A $\TLQS$-formula is said
to be {\em satisfiable\/} if it has a $\TLQS$-model.  A
$\TLQS$-formula is \emph{valid} if it is satisfied by all
$\TLQS$-interpretations.
\subsection{Characterizing $\TLQSR$}\label{restrictionquant}
$\TLQSR$ is the subcollection of the formulae $\psi$ of $\TLQS$ such
that, for \emph{every} atomic formula $(\forall Z_1),\ldots,(\forall
Z_m)\varphi_1$ of level 2 occurring in $\psi$ and \emph{every} level
$1$ atomic formula of the form $(\forall z_1) \ldots (\forall z_n)
\varphi_0$ occurring in $\varphi_1$, the
condition
\begin{equation}
    \label{condition}
\neg \varphi_0 \rightarrow \bigwedge_{i=1}^n \bigvee_{j=1}^m  z_i \in Z_j
\end{equation}
is a valid $\TLQS$-formula (in this case we say that the atom
$(\forall z_1) \ldots (\forall z_n) \varphi_0$ is {\em linked} to the
variables $Z_1, \ldots , Z_m$).

Condition~(\ref{condition}) guarantees that, if a given interpretation
assigns to $z_1,\ldots,z_n$ elements of the domain that make
$\varphi_0$ false, such values are contained in at least one of the
subsets of the domain assigned to $Z_1,\ldots,Z_m$.  As shown in the
proof of statement (ii) of Lemma \ref{quantifiedform}, this fact is
used to make sure that satisfiability is preserved in the finite
model.  As the examples in Section \ref{expressiveness} illustrate,
condition~(\ref{condition}) is not particularly restrictive.

The following question arises: how one can establish whether a given
$\TLQS$-formula is a $\TLQSR$-formula?  Observe that
condition~(\ref{condition}) involves no collection variables and no
quantification.  Indeed, it turns out to be a $2LS$-formula and
therefore one could use the decision procedures in \cite{FerOm1978} to
test its validity, since $\TLQS$ is a conservative extension of $2LS$.
We mention also that in most cases of interest, as will be shown
in detail in Section \ref{expressiveness}, condition~(\ref{condition}) is just an instance of the
simple propositional tautology $\neg( A \rightarrow B)
\rightarrow A$, and therefore its validity can be established just by
inspection.

Finally, we observe that though the semantics of $\TLQSR$ plainly
coincides with the one given above for $\TLQS$-formulae, nevertheless
we will refer to $\TLQS$-interpretations of $\TLQSR$-formulae as
\emph{$\TLQSR$-interpretations}.

\section{Relativized interpretations}\label{machinery}
We introduce the notion of relativized interpretation, to be used
together with the decision procedure of Section \ref{decisionproc}
to construct, out of a model $\model = (D,M)$ for a $\TLQSR$-formula
$\psi$, a finite interpretation $\model^* = (D^*,M^*)$ of bounded
size satisfying $\psi$ as well.

%
\begin{definition}\label{relintrp}
Let $\model=(D,M)$ be a $\TLQSR$-interpretation. Let $D^{*}
\subseteq D$, $d^{*} \in D^{*}$, and ${\cal V}'_{1} \subseteq {\cal
V}_{1}$. The {\em relativized\/} interpretation
$\Rel(\model, D^{*}, d^{*}, {\cal V}'_{1})$ of $\model$ with respect
to $D^{*}$, $d^{*}$, and ${\cal V}'_{1}$ is the interpretation
$(D^{*},M^{*})$ such that
\begin{eqnarray*}
    M^{*}x & = & \left\{
    \begin{array}{ll}
        Mx\,, & \mbox{if $Mx \in D^{*}$}  \\
        d^{*}\,, & \mbox{otherwise}
    \end{array}
    \right.
    \\
    M^{*}X & = & MX \cap D^{*}
    \\
M^{*}A  &=&
    ((MA \cap \pow(D^{*})) \\
        & & \;\;\;\setminus \{M^{*}X: X \in {\cal V}'_{1}\})
    \cup \{M^{*}X: X \in {\cal V}'_{1},~MX \in MA\}\,.
\end{eqnarray*}\eod
\end{definition}
The definition of relativized interpretation given above is inspired
by the construction of the finite model described in
\cite{FerOm1978}, \cite{CanCut90}, and in \cite{CanCut93}. We spend
some words on the intuition behind the definition of $M^{*}A$.
Analogously to $M^{*}X$, $M^{*}A$ is obtained from the intersection
of the interpretation of $A$ in $\model$ with the power set of the
finite domain $D^*$. However, such operation may leave in $MA \cap
\pow(D^{*})$ some sets $J$ such that $J = M^*X$ but $MX \notin MA$.
Such $J$'s have to be removed from the restricted interpretation of
$A$ in order to guarantee that satisfiability of $\psi$ is
preserved.
Further, there also may be some $MX \in MA$ such that $M^{*}X\notin
MA \cap \pow(D^{*})$. Again, to let the restricted model preserve
satisfiability of $\psi$, such $M^{*}X$ have to be added to the
interpretation of $A$ in the restricted model. 

For ease of notation, we will often omit the reference to the
element $d^{*} \in D^{*}$ and write simply\/ $\Rel(\model, D^{*},
{\cal V}'_{1})$ in place of\/ $\Rel(\model, D^{*}, d^{*}, {\cal
V}'_{1})$.


The following satisfiability result holds for unquantified atomic
formulae.

\begin{lemma}
\label{le_basic} Let $\model=(D,M)$ be a $\TLQSR$-interpretation.
Also, let $D^{*} \subseteq D$, $d^{*} \in D^{*}$, and ${\cal V}'_{1}
\subseteq {\cal V}_{1}$ be given. Let us put $\model^{*} =
\Rel(\model, D^{*}, d^{*}, {\cal V}'_{1})$. Then the following
holds.
\begin{itemize}
    \item[(a)] $\model^{*} \models x = y$ ~~iff~~ $\model \models x = y$,
    for all $x,y \in {\cal V}_{0}$ such that $Mx, My \in D^{*}$;

    \item[(b)] $\model^{*} \models x \in X$ ~~iff~~ $\model \models x \in X$,
    for all $X \in {\cal V}_{1}$ and $x \in {\cal V}_{0}$ such that
    $Mx \in D^{*}$;

    \item[(c)] $\model^{*} \models X = Y$ ~~iff~~ $\model \models X = Y$,
    for all $X,Y \in {\cal V}_{1}$ such that
    if $MX \neq MY$ then
    $(MX \Delta MY) \cap D^{*} \neq \emptyset$;

    \item[(d)] if for all $X,Y \in {\cal V}'_{1}$ such that $MX \neq MY$
    we have $(MX \Delta MY) \cap D^{*} \neq \emptyset$, then
    $\model^{*} \models X \in A$ ~~iff~~ $\model \models X \in A$,
    for all $X \in {\cal V}'_{1}$, $A \in {\cal V}_{2}$.\footnote{We
    recall that $\Delta$ denotes the symmetric difference operator
    defined by $s \Delta t = (s \setminus t) \cup (t \setminus s)$.}
\end{itemize}
\end{lemma}
{\bf Proof.}

Cases (a), (b) and (c) are easily verified. We prove only case (d).
%
%
%
To this end, assume that for all $X,Y \in {\cal V}'_{1}$ such that $MX
\neq MY$ we have $(MX \Delta MY) \cap D^{*} \neq \emptyset$.  Let $X
\in {\cal V}'_{1}$ and $A \in {\cal V}_{2}$.  If $MX \in MA$, then
obviously $M^{*}X \in M^{*}A$.  On the other hand, if $MX \notin MA$,
but $M^{*}X \in M^{*}A$, then we must necessarily have $M^{*}X =
M^{*}Z$, for some $Z \in {\cal V}'_{1}$ such that $MZ \in MA$.  But
then, as $MX \neq MZ$, from our hypothesis we would obtain $M^{*}X
\neq M^{*}Z$, which is a contradiction.  \qed

\subsection{Relativized interpretations and quantified atomic formulae}

Satisfiability results for quantified atomic formulae are treated as
shown in the following. Let us put

\begin{eqnarray*}
    \model^{z,*} & = & \Rel(\model^{z},D^{*},{\cal V}'_{1})  \\
    \model^{*,z} & = & \model^{*}[z_{1}/u_{1},\ldots,z_{n}/u_{n}] \\
    \model^{Z,*} & = & \Rel(\model^{Z},D^{*},{\cal V}'_{1} \cup
                         \{Z_{1},\ldots,Z_{m}\})  \\
    \model^{*,Z} & = & \model^{*}[Z_{1}/U_{1},\ldots,Z_{m}/U_{m}]. \\
\end{eqnarray*}

The following lemmas provide useful technical results to be employed
in the proof of Theorem \ref{correctness} below.  In particular,
Lemmas \ref{le_M*zMz*} and \ref{le_eqM*ZMZ*}, which are simply stated
without proof, are used to prove Lemma \ref{quantifiedform}.

\begin{lemma}
\label{le_M*zMz*} Let $u_{1},\ldots,u_{n} \in D^{*}$ and let
$z_{1},\ldots,z_{n} \in {\cal V}_{0}$. Then, for every $x \in {\cal
V}_{0}$ and $X \in {\cal V}_{1}$ we have:
\begin{itemize}
\item [(i)] $M^{*,z}x = M^{z,*}x$,
\item [(ii)] $M^{z,*}X = M^{*,z}X$.
\end{itemize}
%
\end{lemma}
%

\begin{lemma}
\label{le_eqM*ZMZ*} Let $\model=(D,M)$ be a $\TLQSR$-interpretation,
$D^{*} \subseteq D$, ${\cal V}'_{1} \subseteq {\cal V}_{1}$,
$Z_{1},\ldots,Z_{m} \in {\cal V}_{1} \setminus {\cal V}'_{1}$,
$U_{1},\ldots,U_{m} \in \pow(D^*) \setminus \{M^{*}X : X \in {\cal
V}'_{1}\}$.

Then the $\TLQSR$-interpretations $\model^{*,Z}$ and $\model^{Z,*}$
coincide.
\end{lemma}
\begin{lemma}\label{quantifiedform}
Let $\model=(D,M)$ be a $\TLQSR$-interpretation.  Let $D^{*} \subseteq
D$, $d^{*} \in D^{*}$, ${\cal V}'_{1} \subseteq {\cal V}_{1}$ be
given, let $\model^{*} = \Rel(\model, D^{*}, d^{*}, {\cal
V}'_{1})$.  Further, let $(\forall z_1) \ldots (\forall z_n)
\varphi_0$ and $(\forall Z_1) \ldots (\forall Z_m) \varphi_1$ be
atomic formulae of level 1 and 2, respectively, such that $Mx \in
D^*$, for every $x \in {\cal V}_{0}$ occurring in $\varphi_0$ or in
$\varphi_1$.  Then we have
\begin{itemize}
\item [(i)] if $\model\models (\forall z_1) \ldots (\forall z_n)
\varphi_0$, then $\model^*\models (\forall z_1) \ldots (\forall z_n)
\varphi_0$;

\item [(ii)] if $\model\models (\forall Z_1) \ldots (\forall Z_m)
\varphi_1$, then $\model^*\models (\forall Z_1) \ldots (\forall Z_m)
\varphi_1$, provided that
\begin{itemize}
\item $(MX \Delta MY) \cap D^{*} \neq \emptyset$, for every $X,Y \in
{\cal V}_1$ with $MX \neq MY$, and that
\item there are
$u_1,\ldots,u_n \in D^*$ such that
$\model[v_1/u_1,\ldots,v_n/u_n]\not\models \varphi_0$, for every
$(\forall z_1) \ldots (\forall z_n) \varphi_0$ not satisfied by
$\model$, and occurring in
${\varphi_1}_{X_1,\ldots,X_m}^{Z_1,\ldots,Z_m}$, with
$X_1,\ldots,X_m$ variables in ${\cal V}_1'$.
\end{itemize}
\end{itemize}
\end{lemma}
{\bf Proof.}
\begin{itemize}
\item [(i)]
Assume by contradiction that there exist $u_1,\ldots,u_n \in D^*$
such that $\model^{*,z}\not\models \varphi_0$.
Then, there must be an atomic formula $\varphi_0'$ in $\varphi_0$
(either of type $x=y$ or $x \in X$) that is differently interpreted
in $\model^{*,z}$ and in $\model^z$.

Let us suppose first that $\varphi_0'$ is the atom $x = y$ and,
without loss of generality, that $\model^{*,z}\not\models x = y$.
By Lemma \ref{le_M*zMz*}, we have $M^{z,*}x \neq M^{z,*}y$. Since
$M^{z,*}x = M^{z}x$, $M^{z,*}y = M^{z}y$, and, by hypothesis,
$M^{z}x =M^{z}y$, we obtain a contradiction.

Now let us suppose that $\varphi_0$ is the atom $x \in X$ and, without
loss of generality, assume that $\model^{*,z}\not\models x \in X$.
By Lemma \ref{le_M*zMz*}, we have $M^{z,*}x \notin M^{z,*}X$, that
is $M^{z}x \notin M^{z}X \cap D^*$, again a contradiction.
%
%

\item [(ii)] 
%
Assume, by way of contradiction, that $\model^* \not\models(\forall
Z_1) \ldots (\forall Z_m) \varphi_1$. Hence there exist $U_1,
\ldots, U_m \in pow(D^*)$ such that ${\cal
M}^{*,Z}\not\models\varphi_1$.

Without loss of generality, assume that $U_i = M^*X_i$, for $1 \leq
i \leq k$ ($k \geq 0$) for some variables $X_1,\ldots,X_k$ in ${\cal
V}_1'$, and that $U_j \neq M^*X$ for all $k+1 \leq j \leq m$ and for
all variables $X$ in ${\cal V}_1'$.

Let $\bar{\varphi}_1$ be the formula obtained by simultaneously
substituting $Z_1,\ldots,Z_k$ with $X_1,\ldots,X_k$ in $\varphi_1$,
and let $\model^{*,Z_k} =
\model^*[Z_{k+1}/U_{k+1},\ldots,Z_{m}/U_m]$. Further, let ${\cal
M}^{Z'}$ be a $\TLQSR$-interpretation differing from $\model^Z$ only
in the evaluation of $Z_1,\ldots,Z_k$ ($M^{Z'}Z_1 = MX_1,\ldots,
M^{Z'}Z_k = MX_k$).

Now we can distinguish two cases.

\medskip

If $k =m$, then $\model^{*,Z_k}$ and $\model^*$ coincide and a
contradiction can be obtained by showing that the implications

$$\model^{*,Z} \not\models  \varphi_1 \Rightarrow \model^*\not\models \bar{\varphi}_1 \Rightarrow \model \not\models \bar{\varphi}_1 \Rightarrow \model^{Z'}\not\models \varphi_1$$

hold. Hence, 
against the hypothesis, we get that $\model \not\models (\forall
Z_1)\ldots(\forall Z_m)\varphi_1$. The first implication,
$\model^{*,Z} \not\models  \varphi_1 \Rightarrow \model^*\not\models
\bar{\varphi}_1$, is plainly derived from the definition of
$\bar{\varphi}_1$. The second one, $\model^*\not\models
\bar{\varphi}_1 \Rightarrow \model \not\models \bar{\varphi}_1$, can
be proved by showing that $\model^*$ and $\model$ interpret each
atomic formula $\bar{\varphi}_1'$ occurring in $\bar{\varphi}_1$ in
the same manner.

If $\bar{\varphi}_1'$ is an atomic formula of level 0 or an atomic
formula of level 1 of type $X = Y$ and $X \in A$, the proof follows
directly from Lemma \ref{le_basic}.

If $\bar{\varphi}_1'$ is an atomic formula of level 1 of type
$(\forall z_1)\ldots,(\forall z_n)\varphi_0$, the implication
$\model\models (\forall z_1) \ldots (\forall z_n) \varphi_0$
$\Rightarrow$ $\model^*\models (\forall z_1) \ldots (\forall z_n)
\varphi_0$ follows from statement (i) of the lemma, whereas
$\model^*\models (\forall z_1) \ldots (\forall z_n) \varphi_0$
$\Rightarrow$ $\model\models (\forall z_1) \ldots (\forall z_n)
\varphi_0$ can be proved by contradiction as follows.  Assume that
$\model \not\models (\forall z_1) \ldots (\forall z_n) \varphi_0$.
Then, by hypothesis, there are $u_1,\ldots,u_n$ in $D^*$ such that
$\model[z_1/u_1,\ldots,z_n/u_n]\not\models \varphi_0$ and, by Lemmas \ref{le_basic},
\ref{le_M*zMz*}, $\model^*[z_1/u_1,\ldots,z_n/u_n]\not\models
\varphi_0$ contradicting our hypothesis.

The last implication, $\model \not\models \bar{\varphi}_1
\Rightarrow \model^{Z'}\not\models \varphi_1$, is deduced by the
definition of $\bar{\varphi}_1$ and of $Z'$.

\medskip

If $k < m$, the schema of the proof is analogous to the previous case.
However, since $\model^{*,Z_k}$ and $\model^*$ do not coincide, the
single steps are carried out in a slightly different manner.  Thus,
for the sake of clarity we report the proof below.

In order to obtain a contradiction we prove that

$$\model^{*,Z} \not\models  \varphi_1 \Rightarrow \model^{*,Z_k}\not\models \bar{\varphi}_1 \Rightarrow \model^{Z_k}\not\models \bar{\varphi}_1 \Rightarrow \model^{Z'}\not\models \varphi_1$$

hold. 

The first implication $\model^{*,Z} \not\models \varphi_1 \Rightarrow
\model^{*,Z_k} \not\models \bar{\varphi}_1$ can be immediately
deduced from the definition of $\bar{\varphi}_1$ and of
$\model^{*,Z_k}$.  The second implication $\model^{*,Z_k}\not\models
\bar{\varphi}_1 \Rightarrow \model^{Z_k}\not\models \bar{\varphi}_1$
can be proved by showing that every atomic formula $\bar{\varphi}_1'$
in $\bar{\varphi}_1$ is interpreted in $\model^{*,Z_k}$ and in
$\model^{Z_k}$ in the same way.

The proof is straightforwardly carried out using Lemmas
\ref{le_eqM*ZMZ*} and \ref{le_basic} in case $\bar{\varphi}_1'$ is
an atomic formula of level 0, or an atomic formula of level 1 of
type $X = Y$ and $X \in A$.

If $\bar{\varphi}_1'$ is an atomic formula of level 1 of type
$(\forall z_1)\ldots(\forall z_n)\varphi_0$,
%
we first show that $\model^{Z_k} \models (\forall z_1)\ldots(\forall
z_n)\varphi_0$ implies that $\model^{*,Z_k} \models (\forall
z_1)\ldots(\forall z_n)\varphi_0$.

If $\model^{Z_k} \models (\forall z_1)\ldots(\forall z_n)\varphi_0$,
we have that $\model^{Z_k,*}
\models (\forall z_1)\ldots(\forall z_n)\varphi_0$,
by (i) of the present lemma, and that $\model^{*,Z_k} \models (\forall
z_1)\ldots(\forall z_n)\varphi_0$, by Lemma
\ref{le_eqM*ZMZ*}.

Now, let us show that $\model^{*,Z_k} \models (\forall
z_1)\ldots(\forall z_n)\varphi_0$ implies that $\model^{Z_k} \models
(\forall z_1)\ldots(\forall z_n)\varphi_0$.

Assume by contradiction that $\model^{Z_k} \not\models (\forall
z_1)\ldots(\forall z_n)\varphi_0$.  Then there exist $u_1,\ldots,u_n
\in D$ such that $\model^{Z_k}[z_1/u_1,\ldots,z_n/u_n]
\not\models\varphi_0$.  In particular, by the condition $\neg
\varphi_0 \rightarrow \And_{i = 1}^n\Or_{j = k+1}^{m} z_i \in Z_j$,
we derive that $u_1,\ldots,u_n$ are elements of $D^*$.  This allows us
to apply Lemma \ref{le_basic} and to obtain that
$\model^{Z_k,z,*}\not\models\varphi_0$.  Then, by Lemma
\ref{le_M*zMz*} we obtain $\model^{Z_k,*,z}\not\models\varphi_0$ and
hence $\model^{Z_k,*}\not\models(\forall z_1)\ldots(\forall
z_n)\varphi_0$. Thus, Lemma \ref{le_eqM*ZMZ*} yields
$\model^{*,Z_k}\not\models (\forall z_1)\ldots(\forall z_n)\varphi_0$
contradicting the hypothesis.

Finally, the third implication, $\model^{Z_k}\not\models
\bar{\varphi}_1 \Rightarrow \model^{Z'}\not\models \varphi_1$ is
directly derived from the definition of $\bar{\varphi}_1$ and of
$Z'$. \qed
%
%
\end{itemize}

\section{The satisfiability problem for
$\TLQSR$-formulae}\label{satisfiability}
In this section we solve the satisfiability problem for $\TLQSR$,
i.e. the problem of establishing for any given formula of $\TLQSR$
whether it is satisfiable or not, as follows:
\begin{itemize}
    \item[(a)] firstly, we reduce effectively the satisfiability
    problem for $\TLQSR$-formulae to the satisfiability problem for
    normalized $\TLQSR$-conjunctions (these will be defined precisely
    below);

    \item[(b)] secondly, we prove that the collection of normalized
    $\TLQSR$-conjunctions enjoys a small model property.
\end{itemize}
From (a) and (b), the solvability of the satisfiability problem for
$\TLQSR$ follows immediately.  Additionally, by further elaborating on
point (a), it could easily be shown that indeed the whole collection
of $\TLQSR$-formulae enjoys a small model property.


\subsection{Normalized $\TLQSR$-conjunctions}\label{normal3LQS}
Let $\psi$ be a formula of $\TLQSR$ and let $\psi_{DNF}$ be a
disjunctive normal form of $\psi$.  Then $\psi$ is satisfiable if
and only if at least one of the disjuncts of $\psi_{DNF}$ is
satisfiable. We recall that the disjuncts of $\psi_{DNF}$ are
conjunctions of level $0,1$, and $2$ literals, i.e. level $0,1$, and
$2$ atoms or their negation.  In view of the previous observations,
without loss of generality, we can suppose that our formula $\psi$
is a conjunction of level $0,1$, and $2$ literals.  In addition, we
can also assume that no bound variable in $\psi$ can occur in more
than one quantifier or can occur also free.

For decidability purposes, negative quantified conjuncts occurring
in $\psi$  can be eliminated as explained below. Let $\model = (D,
M)$ be a model for $\psi$. Then $\model \models \neg (\forall z_1)
\ldots (\forall z_n) \varphi_0$ if and only if ${\cal
M}[z_1/u_1,\ldots , z_n/u_n] \models \neg \varphi_0$, for some
$u_1,\ldots ,u_n \in D$, and $\model \models \neg (\forall Z_1)
\ldots (\forall Z_m) \varphi_1$ if and only if\linebreak
$\model[Z_1/U_1,\ldots , Z_m/U_m] \models \neg \varphi_1$, for some
$U_1,\ldots , U_m \in \pow(D)$.
Thus, each negative literal of type $\neg (\forall z_1) \ldots
(\forall z_n) \varphi_0$ can be replaced by $\neg
(\varphi_{0})^{z_{1},\ldots,z_{n}}_{z'_{1},\ldots,z'_{n}}$, where
$z'_{1},\ldots,z'_{n}$ are newly introduced variables of sort 0.
Likewise, each negative literal of type $\neg (\forall Z_1) \ldots
(\forall Z_m) \varphi_1$ can be replaced by $\neg
(\varphi_{1})^{Z_{1},\ldots,Z_{m}}_{Z'_{1},\ldots,Z'_{m}}$, where
$Z'_{1},\ldots,Z'_{m}$ are newly introduced variables of sort 1.


Hence, we can further assume that $\psi$ is a conjunction of
literals of the following types:
\begin{itemize}
\item[(1)] $x=y$, $\neg (x = y)$, $x \in X$, $\neg (x \in X)$, $X=Y$, $\neg(X
=Y)$, $X \in A$, $\neg(X \in A)$;

\item[(2)] $(\forall z_1) \ldots (\forall
z_n) \varphi_0$, where $n > 0$ and $\varphi_0$ is a propositional
combination of level $0$ atoms;

\item[(3)] $(\forall Z_1) \ldots
(\forall Z_m) \varphi_1$, where $m > 0$ and $\varphi_1$ is a
propositional combination of level $0$ and level $1$ atoms, where
atoms of type $(\forall z_1) \ldots (\forall z_n) \varphi_0$ in
$\varphi_{1}$ are linked to the bound variables $Z_1, \ldots , Z_m$.
\end{itemize}

We call such formulae {\em normalized $\TLQSR$-conjunctions}.

\subsection{A small model property for normalized $\TLQSR$-conjunctions}
\label{decisionproc}
In view of the preceding discussion we can limit ourselves to
consider the satisfiability problem for normalized
$\TLQSR$-conjunctions only.

Thus, let $\psi$ be a normalized $\TLQSR$-conjunction and assume
that $\model = (D, M)$ is a model for $\psi$.

We show how to construct, out of $\model$, a finite
$\TLQSR$-interpretation $\model^* = (D^*,M^*)$ which is a model of
$\psi$.  We proceed as follows.  First we outline a procedure to build
a nonempty finite universe $D^* \subseteq D$ whose size depends solely
on $\psi$ and can be computed {\em a priori}.  Then, a finite
$\TLQSR$-interpretation $\model^* = (D^*,M^*)$ is constructed
according to Definition \ref{relintrp}.  Finally, we show that
$\model^*$ satisfies $\psi$.

\subsubsection{Construction of the universe $D^*$.}\label{ssseUniv}
Let us denote by ${\cal W}_0$, ${\cal W}_1$, and ${\cal W}_2$ the
collections of the variables of sort $0$, $1$, and $2$ present in
$\psi$, respectively. Then we compute $D^*$ by means of the
procedure described below.

%

Let $\psi_1, \ldots , \psi_k$ be the conjuncts of $\psi$. To each
conjunct $\psi_i$ of the form $(\forall Z_{i1}) \ldots (\forall
Z_{im_i}) \varphi_i$ we associate the collection $\varphi_{i1},
\ldots , \varphi_{i\ell_i}$ of atomic formulae of type (2) present
in the matrix of $\psi_i$ and call the variables $Z_{i1}, \ldots ,
Z_{im_i}$ the {\em arguments of } $\varphi_{i1}, \ldots ,
\varphi_{i\ell_i}$. Then we put
$$
\Phi = \{ \varphi_{ij} : 1 \leq i \leq k \hbox{ and } 1 \leq j \leq
\ell_i \}.
$$

For every pair of variables $X,Y$ in ${\cal W}_1$ such that $MX \neq
MY$, let $u_{X,Y}$ be any element in the symmetric difference of
$MX$ and $MY$ and put $\Delta_{1} = \{u_{XY}: X,Y \hbox{ in } {\cal
W}_1
\hbox{ and } MX \neq MY\}$. If $\Delta_1$ is constructed applying
the procedure \emph{Distinguish} described in \cite{CanFer1995}, it holds that
$|\Delta_1| \leq |{\cal W}_1| -1$.


 We initialize $D^*$ with the set $\{Mx : x \hbox{ in } {\cal W}_0
\} \cup \Delta_{1}$. Then, for each $\varphi \in \Phi$ of the form
$(\forall z_1) \ldots (\forall z_n) \varphi_0$ having
$Z_1,\ldots,Z_m$ as arguments and for each ordered $m$-tuple
$(X_{i_1},\ldots,X_{i_m})$ of variables in ${\cal W}_1$, if ${\cal
M}{\varphi_0}_{X_{i_1},\ldots,X_{i_m}}^{Z_1\;\,,\ldots,\;Z_m} =
\hbox{\bf false}$ we insert in $D^*$ elements $u_1,\ldots,u_n \in D$
such that
$$\model[z_1/u_1,\ldots,z_n/u_n]
{\varphi_0}_{X_{i_1},\ldots,X_{i_m}}^{Z_1\;,\ldots,\:Z_m}
=\false\,,$$ otherwise we leave $D^*$ unchanged.

From the previous construction it easily follows that
\begin{equation}\label{DStar}
|D^*| \leq |{\cal W}_0| + |{\cal W}_1| -1 + ((|{\cal W}_1|^{\max m})\max n)|\Phi|\, ,
\end{equation}
where $\max m$ and $\max n$ are respectively the maximal number of quantifiers in formulae of level 2 and the maximal number of
quantifiers in formulae of level 1 occurring in quantified formulae of level 2. Thus, in general, the domain of the small model $D^*$ is exponential in the size of the input formula $\psi$.

\subsubsection{Correctness of the relativization.}
Let us put $\model^* = \Rel(\model, D^*,{\cal W}_1)$. We have to
show that, if $\model\models \psi$, then $\model^*\models\psi$.

\begin{theorem}\label{correctness}
Let $\model$ be a $\TLQSR$-interpretation satisfying a normalized
$\TLQSR$-conjunction $\psi$. Further, let $\model^* = \Rel(\model,
D^*,{\cal W}_1)$ be the $\TLQSR$-interpretation defined according to
Definition \ref{relintrp}, where $D^*$ is constructed as described
above, and let ${\cal W}_1$ be defined as above. Then
$\model^*\models \psi$.
\end{theorem}
{\bf Proof.}

We have to prove that $\model^*\models \psi'$ for every literal
$\psi'$ in $\psi$. Each $\psi'$ is of one of the three types
introduced in Section \ref{normal3LQS}. By applying Lemma
\ref{le_basic} or \ref{quantifiedform} to every $\psi'$ in $\psi$
(according to the type of $\psi'$) we obtain the thesis.

Notice that the hypotheses of Lemma \ref{le_basic} and of Lemma
\ref{quantifiedform} are fulfilled by the construction of $D^*$
outlined above:
\begin{itemize}
\item $Mx \in D^*$, for every variable $x\in {\cal V}_0$.
Furthermore, $(MX \Delta MY) \cap D^* \neq \emptyset$ for every
$X,Y\in {\cal V}_1$ such that $MX \neq MY$ (one just needs to
substitute the generic set of individual variables ${\cal V}_0$ with
${\cal W}_0$ and ${\cal V}_1$ with ${\cal W}_1$);

\item for every atomic formula of type $(\forall z_1)\ldots(\forall
z_n)\varphi_0$ occurring in an atomic formula of level 2 and such
that $\model\not\models(\forall z_1)\ldots(\forall z_n)\varphi_0$,
there are $u_1,\ldots,u_n$ elements of $D^*$ such that
$\model[z_1/u_1,\ldots,z_n/u_n]\not\models \varphi_0$. \qed
%
\end{itemize}

From the above reduction and relativization steps, it is not hard to
derive the following result:
\begin{corollary}
    The fragment $\TLQSR$ enjoys a small model property (and
    therefore its satisfiability problem is solvable). \qed
\end{corollary}

\section{Expressiveness of the language $\TLQSR$}\label{expressiveness}
 Several constructs of elementary set theory are
easily expressible within the language $\TLQSR$.  In particular, as
will be shown below, it is possible to express with $\TLQSR$-formulae
a restricted variant of the set former, which in turn allows to
express other significant set operators such as binary union,
intersection, set difference, the singleton operator, the powerset
operator and its variants, etc.

More specifically, atomic formulae of type $X=\{z : \varphi(z)\}$ can
be expressed in $\TLQSR$ by the formula

\begin{equation}
\label{eq_setformX} (\forall z)(z \in X \leftrightarrow \varphi(z))\,,
\end{equation}
provided that after transforming it into prenex normal form one
obtains a formula satisfying the syntactic constraints of $\TLQSR$.
In particular, this is always the case whenever $\varphi(z)$ is any
unquantified formula of $\TLQSR$.

The same remark applies also to atomic formulae of type $A=\{Z :
\varphi(Z)\}$.  In this case, in order for a prenex normal form of
\begin{equation}
\label{eq_setformA} (\forall Z)(Z \in A \leftrightarrow \varphi(Z))
\end{equation}
to be in the language $\TLQSR$, it is enough that
\begin{itemize}
\item[(a)] $\varphi(Z)$ does not contain
any quantifier over variables of sort $1$, and
\item[(b)] all quantified variables of
sort $0$ in $\varphi(Z)$ are linked to the variable $Z$ as specified in condition~(\ref{condition}).
\end{itemize}

In what follows we introduce the stratified syllogistics $2LS$,
already mentioned in the introduction, and $3LSSP$ (Three-Level
Syllogistic with Singleton and Powerset), and describe their
formalization in $\TLQSR$. Then we show how to express some other
set-theoretical constructs. Finally, in Section \ref{sec:otherappl} we introduce a family of sublanguages
of $\TLQSR$ having the satisfiability problem \textsf{NP}-complete and able to express the modal logic $\Sc$.
%
%


\subsection{Two-Level Syllogistic} $2LS$ is a fragment
of the elementary theory of sets admitting individual variables,
$x,y, z,\ldots$, set variables, $X,Y,Z,\ldots$, and the constants
$0$ and $1$ standing respectively for the empty set and the domain
of the discourse. Terms and formulae of $2LS$ are constructed out of
variables and constants by means of the set operators
of union, intersection, and set complementation, the binary relators
$=$, $\in$, and $\subseteq$, and the propositional connectives.

\begin{table}[t]
\begin{center}
\begin{tabular}{|l c l|}
\hline
& &  \\
 \hspace*{0.3cm}   (1) $X = Y$, $X \subseteq Y$ & \hspace*{1cm} &
(5) $\neg(X = Y)$, $\neg(X \subseteq Y)$  \hspace*{0.3cm}\\
 \hspace*{0.3cm}   (2) $X = Y \cap Z$, $X = Y \cup Z$ & \hspace*{1cm} & (6) $x \in X$  \hspace*{0.3cm}\\

 \hspace*{0.3cm}   (3) $X = \compl{Y}$ &             & (7) $\neg(x = y)$ \hspace*{0.3cm}\\
 \hspace*{0.3cm}   (4) $X = 0$, $X = 1$ &          & (8) $x = y$  \hspace*{0.3cm}\\
    & &  \\
\hline
\end{tabular}
\end{center}
\caption{$2LS$ literals.}\label{2LSliterals}
\end{table}

$2LS$ has been proved decidable in \cite{FerOm1978} by a
procedure that, taking as input a conjunction $\varphi$ of literals
of the forms illustrated in Table~\ref{2LSliterals}, stops with
failure in case $\varphi$ is unsatisfiable, otherwise returns a
model for $\varphi$.


Every literal from Table~\ref{2LSliterals} can readily be expressed as
a formula of $\TLQSR$.  Indeed, $X = Y$ is an atomic formula of level
1 of $\TLQSR$, whereas $X \subseteq Y$ can be expressed by the
quantified atomic formula $(\forall z)(z \in X \rightarrow z \in Y)$
of level $1$.  $X = Y \cup Z$ can be translated into the formula
$(\forall z) ((z \in Y \Or z \in Z) \leftrightarrow z \in X)$ and $X =
Y \cap Z$ into $(\forall z) ((z \in Y \And z \in Z) \leftrightarrow z
\in X)$.  $X = \compl{Y}$ can be expressed by $(\forall z)(z \in X
\leftrightarrow \neg (z\in Y))$.  Literals of type $X = 0$ and $X = 1$
are translated in the atomic formulae of level 1 $(\forall z) \neg(z
\in X)$ and $(\forall z)(z\in X)$,
respectively. 
Literals of $2LS$ of type $(6), (7),$ and $(8)$ are just atomic
formulae of $\TLQSR$ of level 0.

\begin{table}[t]
\begin{center}
\begin{tabular}{|l c l|}
\hline
& &  \\
 \hspace*{0.3cm}   (1) $A = B$, $A \subseteq B$ & \hspace*{1cm} \hspace*{0.3cm}&
(4) $A = 0$, $A = 1$ \\
 \hspace*{0.3cm}   (2) $A = B \cap C$, $A = B \cup C$ & \hspace*{1cm} & (5) $\neg(A = B)$, $\neg (A \subseteq B)$  \hspace*{0.3cm}\\

 \hspace*{0.3cm}   (3) $A = \compl{A}$ &   & (6) $X \in A$ \hspace*{0.3cm}\\
& &  \\
\hline
\end{tabular}
\end{center}
\caption{$3LS$ literals.}\label{3LSliterals}
\end{table}

\begin{table}[t]
\begin{center}
\begin{tabular}{|l c c c l|}
\hline
 & & & &  \\
  \hspace*{0.3cm}  (1) $X = \{x\}$ & \hspace*{0.4cm} &(2) $A = \{X\}$& \hspace*{0.4cm} &(3) $A = \pow(X)$ \hspace*{0.3cm} \\
 & & & & \\
\hline
\end{tabular}
\end{center}
\caption{Literals with singleton and powerset set
operators.}\label{SingPowliterals}
\end{table}

\subsection{Three Level Syllogistic with Singleton and Powerset}
$3LSSP$ is the sublanguage of $3LSSPU$ not involving the set
theoretic construct of general union.  It can be obtained from $2LS$
by extending it with a new sort of variables $A, B, C,\ldots$,
ranging over collections of sets. Furthermore, besides the usual set
theoretical constructs, $3LSSP$ involves the set singleton operator
$\{\cdot\}$ and the powerset operator $\pow$.

$3LSSP$ can plainly be decided by the decision procedure presented
in \cite{CanCut93} for the whole $3LSSPU$.

All formulae in Tables \ref{3LSliterals} and \ref{SingPowliterals}
are readily expressible by $\TLQSR$-formulae.  For instance, $A = B$
can be translated into the $\TLQSR$-formula $(\forall Z)(Z \in A
\leftrightarrow Z \in B)$ of level 2, whereas $A \subseteq B$ can be
formalized as $(\forall Z)(Z \in A \rightarrow Z \in B)$.  The
literals $A = B \cap C$ and $A = B \cup C$ can be translated into
$(\forall Z)(Z \in A \leftrightarrow (Z \in B \And Z \in C))$ and
$(\forall Z)(Z \in A \leftrightarrow (Z \in B \Or Z \in C))$,
respectively.  $A = \compl{B}$ can be expressed by $(\forall Z)(Z
\in A \leftrightarrow \neg(Z \in B))$.  Literals of type $A = 0$ and
$A = 1$ can be expressed by the formulae $(\forall Z) \neg(Z \in A)$
and $(\forall Z)(Z \in
A)$, respectively. 
Literals of type (6) are just atomic $\TLQSR$-formulae of level
1.

The singleton of level 1, $X = \{x\}$, is expressed by the atomic
formula $(\forall z)(z \in X \leftrightarrow z = x)$ of level 1,
whereas the singleton of level 2, $A = \{X\}$, is translated into the
formula $(\forall Z)(Z \in A \leftrightarrow Z = X)$ of level 2.
Finally, the powerset of a set $X$, $A = \pow(X)$, is translated into
the formula $\varphi \equiv (\forall Z)(Z \in A \leftrightarrow
(\forall z) (z \in Z \rightarrow z \in X))$.  It is easy to check that
$\varphi$ satisfies the restriction on quantifiers introduced in
Section \ref{restrictionquant}.  In fact, putting $\varphi_0 \equiv (z
\in Z \rightarrow z \in X)$ and considering that the general
expression $\bigwedge_{i=1}^n \bigvee_{j=1}^m z_i \in Z_j$ in this
case just reduces to $z \in Z$, we have that the condition
$\neg\varphi_0 \rightarrow \bigwedge_{i=1}^n \bigvee_{j=1}^m z_i \in
Z_j$ is instantiated to $\neg(z \in Z \rightarrow z \in X) \rightarrow
z \in Z$, which is an instance of the tautological schema $\neg(A
\rightarrow B)\rightarrow A$.

%

\subsection{Other set theoretical constructs expressible in $\TLQSR$}\label{applsetcomplex} Other
constructs of set theory are expressible in the $\TLQSR$ formalism.

For instance, the literal $A = \pow_{\leq h}(X)$, where $\pow_{\leq
h}(X)$ denotes the collection of all the subsets of $X$ having at
most $h$ distinct elements, can be expressed in $\TLQSR$ as
$$
\begin{array}{ll}
  (\forall Z)(Z \in A \leftrightarrow & ((\forall z)(z \in Z \rightarrow z \in X) \\
                                      & \And (\forall z_1)\ldots
  (\forall z_{h+1})(\bigwedge_{i=1}^{h+1} z_i \in Z \\
                                      & \;\;\;\;\;\rightarrow \neg (\bigwedge_{i=1}^{h+1}\bigwedge_{j=1, j \neq i}^{h+1} \neg(z_i
  = z_j))))).
                                      \\
\end{array}
$$
Likewise, the literals $A = \pow_{< h}(X)$ and $A = \pow_{= h}(X)$, where $\pow_{< h}(X)$ and $\pow_{= h}(X)$ denote, respectively, the
collection of subsets of $X$ with less than $h$ elements and the collection of subsets of $X$ with exactly $h$ distinct elements, can
be expressed in an analogous way.

In the formalization of $A = \pow_{\leq h}(X)$ given above, the
restriction on quantifiers of Section \ref{restrictionquant} is
satisfied. This can easily be verified for both conjuncts
\begin{itemize}
\item [] $\varphi_1 \equiv (\forall z)(z \in Z \rightarrow z \in X)$, and
\item [] $\varphi_2 \equiv (\forall z_1)\ldots
  (\forall z_{h+1})(\bigwedge_{i=1}^{h+1} z_i \in Z \rightarrow \neg (\bigwedge_{i=1}^{h+1}\bigwedge_{j=1, j \neq i}^{h+1} \neg(z_i
  = z_j)))$.
\end{itemize}
The verification of the validity of condition~(\ref{condition}) for
$\varphi_1$ is identical to the one shown for the formula considered
in the previous paragraph. Thus we just check its validity for
$\varphi_2$.

One just needs to put $\varphi_0 \equiv \bigwedge_{i=1}^{h+1} z_i \in
Z \rightarrow \neg (\bigwedge_{i=1}^{h+1}\bigwedge_{j=1, j \neq
i}^{h+1} \neg(z_i = z_j))$ and observe that $\bigwedge_{i=1}^n
\bigvee_{j=1}^m z_i \in Z_j$ is just $\bigwedge_{i=1}^{h+1} z_i \in
Z$.  Thus $\neg\varphi_0 \rightarrow \bigwedge_{i=1}^n \bigvee_{j=1}^m
z_i \in Z_j$ is just the formula
$$
\neg(\bigwedge_{i=1}^{h+1} z_i \in
Z \rightarrow \neg (\bigwedge_{i=1}^{h+1}\bigwedge_{j=1, j \neq
i}^{h+1} \neg(z_i =
  z_j))) \rightarrow \bigwedge_{i=1}^{h+1} z_i \in Z
$$
which again is plainly an instance of the propositional tautology
$\neg(A \rightarrow B)\rightarrow A$.



The Cartesian product $A = X_1 \otimes \ldots \otimes X_n$ can be
formalized by the $\TLQSR$-formula

$$
\begin{array}{ll}
  (\forall Z)(Z \in A \leftrightarrow & ((\forall z_1)\ldots
  (\forall z_{n})(\bigwedge_{i=1}^{n} z_i \in Z \rightarrow \bigwedge_{i=1}^{n} z_i \in X_i)\\
                                      &\;\;\;\;\;\;\And (\forall z_1)\ldots
  (\forall z_{n+1})(\bigwedge_{i=1}^{n+1} z_i \in Z \\
                                      & \;\;\;\;\;\;\;\;\;\;\;\;\rightarrow \neg (\bigwedge_{i=1}^{n+1}\bigwedge_{j=1, j \neq i}^{n+1} \neg(z_i
  = z_j))))).\\
\end{array}
$$

Even in this case, condition~(\ref{condition}) on quantifiers is
satisfied as well.  This can be checked for both the conjuncts
\begin{itemize}
\item [] $(\forall z_1)\ldots
  (\forall z_{n})(\bigwedge_{i=1}^{n} z_i \in Z \rightarrow \bigwedge_{i=1}^{n} z_i \in
  X_i)$, and
\item [] $(\forall z_1)\ldots
  (\forall z_{n+1})(\bigwedge_{i=1}^{n+1} z_i \in Z \rightarrow \neg (\bigwedge_{i=1}^{n+1}\bigwedge_{j=1, j \neq i}^{n+1} \neg(z_i
  = z_j)))$
\end{itemize}
just as in the previous examples.


Another interesting variant of the power set is the
$\pow^*(X_1,\ldots,X_n)$, which denotes the collection
$$\{Z : Z \subseteq \bigcup_{i=1}^n X_i
         \hbox{ and } Z \cap X_i \neq \emptyset,
         \hbox{ for all } 1 \leq i \leq n \}$$
introduced in \cite{Can91}. The literal $A = \pow^*(X_1,\ldots,X_n)$
is expressed in $\TLQSR$ by

$$
\begin{array}{ll}
(\forall Z)(Z \in A \leftrightarrow\\
 \;\;\;\;\;\;\;\;\;\;\;\; ((\forall z)(z \in Z \rightarrow z \in \bigvee_{i=1}^n z \in X_i)
           \wedge \bigwedge_{i=1}^n \neg(\forall z)(z \in Z \rightarrow \neg z \in
           X_i)).
\end{array}
$$

Also with this formula one can verify that the restriction on
quantifiers is satisfied by checking the subformulae:
\begin{itemize}
\item [] $(\forall z)(z \in Z \rightarrow z \in \bigvee_{i=1}^n z \in
X_i)$,
\item [] $(\forall z)(z \in Z \rightarrow \neg z \in X_i)$, for $i =
1,\ldots,n$.
\end{itemize}

\subsection{Other applications of $\TLQSR$}\label{sec:otherappl}
In this section we introduce a family $\{(\TLQSR)^h\}_{h \geq 2}$ of sublanguages of $\TLQSR$, each having the satisfiability problem \textsf{NP}-complete. Then, in Section \ref{sec:modalSc} we illustrate how the modal logic $\Sc$ can be expressed by the language $(\TLQSR)^3$.

Formulae in $(\TLQSR)^h$ must satisfy several syntactic constraints,
as specified in Definition \ref{def:hlang} below, that are crucial to establish
\textsf{NP}-completeness of the satisfiability problem for the language, specifically to show that it is in \textsf{NP}.
First of all, the length of all quantifier prefixes occurring in a formula of
$(\TLQSR)^h$ must be bounded by the constant $h$.
Thus, given a satisfiable $(\TLQSR)^h$-formula $\varphi$ and a
$\TLQSR$-model $\model = (D, M)$ for it, from
Theorem~\ref{correctness} it follows that $\varphi$ is satisfied by
the relativized interpretation $\model^* = (D^{*}, M^{*})$ of $\model$
with respect to a domain $D^{*}$ having its size bounded as specified in
(\ref{DStar}).  Since  $\max m$ and $\max n$ occurring in (\ref{DStar}) are bounded by
the constant $h$, it
follows that the bound expressed in (\ref{DStar}) is polynomial in the size of
$\varphi$.  The other syntactic constraints on
$(\TLQSR)^h$-formulae are introduced to deduce that $M^{*} A \subseteq
\pow_{\pi, h}(D^{*})$, for any free variable $A$ of sort 2 in
$\varphi$, so that the model
$\model^{*}$ can be guessed in nondeterministic polynomial time in the
size of $\varphi$, and the fact that $\model^{*}$ actually satisfies $\varphi$
can be verified in deterministic polynomial time. This is enough to prove that the
satisfiability problem for $(\TLQSR)^h$-formulae is in \textsf{NP}.

\begin{definition}[$(\TLQSR)^h$-formulae]\label{def:hlang}
Let $\varphi$ be a $\TLQSR$-formula and let $A_1,\ldots,A_p$ be all the variables of sort 2 occurring in it. Then $\varphi$ is a $(\TLQSR)^h$-formula, with
$h \geq 2$, if it has the form
\[
\xi_U \And \xi_{\pi, h} \And \psi_1 \And
\ldots \And \psi_p  \And
\chi \, ,
\]
where
\begin{enumerate}
\item $\xi_U \equiv (\forall z)(z \in X_U)$

\emph{i.e., $X_U$ is the (nonempty) universe of discourse};

\item $\xi_{\pi, h} \equiv (\forall Z)\left(Z \in A_{\pi, h}
\leftrightarrow
      (\forall z_1)\ldots (\forall
      z_{h})\left(\bigwedge_{i=1}^{h} z_i \in Z \rightarrow
      \bigvee_{i,j=1, i<j}^{h} z_i = z_j\right)\right)$

\emph{i.e., $A_{\pi, h} = \pow_{< h}(X_U)$ (together with
formula $\xi_U$)};

\item $\psi_i \equiv (\forall Z)(Z \in A_i \rightarrow Z \in
A_{\pi, h})$, for $i = 1,\ldots,p$;

\emph{i.e., $A_i \subseteq \pow_{\pi h}(X_U)$, for $i = 1,\ldots,p$
(together with formulae $\xi_U$ and $\xi_{\pi, h}$)};

\item $\chi$ is a propositional combination of
\begin{enumerate}
    \item quantifier-free atomic formulae of any level,

    \item\label{levelOne} quantified atomic formulae of level 1 of the form
    \[
    (\forall
    z_1)\ldots(\forall z_n)\varphi_0\,,
    \]
    with $n \leq h$,

    \item\label{levelTwo} quantified atomic formulae of level 2 of the form
    \[
    (\forall
    Z_1)\ldots (\forall Z_m)((Z_1 \in A_{\pi, h} \wedge \ldots
    \wedge Z_m \in A_{\pi, h}) \rightarrow \varphi_1)\,,
    \]
    where $m \leq h$ and $\varphi_1$ is a propositional combination of
    quantifier-free atomic formulae and of quantified atomic formulae
    of level 1 satisfying (\ref{levelOne}) above.
\end{enumerate}
\end{enumerate}

\end{definition}

Next we give the following complexity result on $(\TLQSR)^h$.
\begin{theorem}
The satisfiability problem for $(\TLQSR)^h$ is \textsf{NP}-complete.
\end{theorem}
\begin{proof}
\textsf{NP}-hardness of our problem can be proved by reducing an instance of the satisfiability problem for propositional logic to our
problem.

We prove that our problem is in \textsf{NP} reasoning as follows.
Let
\begin{equation}
    \label{eq_QLQSR}
    \varphi \equiv \xi_U \And \xi_{\pi, h} \And \psi_1 \And
\ldots \And \psi_p  \And
\chi
\end{equation}
be a satisfiable $(\TLQSR)^{h}$-formula, and let $\Hf$ be a set of
formulae defined as follows.  Initially, we put
\[
\Hf := \{\xi_U, \xi_{\pi, h}, \psi_1,
\ldots , \psi_p,
\chi\}
\]
and then, we modify $\Hf$ according to the following five rules, until no
rule can be further applied:
\begin{enumerate}[~~R1:]
\item if $\xi \equiv \neg\neg \xi_1$ is in $\Hf$, then $\Hf = (\Hf
\setminus \{\xi\}) \cup \{\xi_1\}$,

\item if $\xi \equiv \xi_1 \And \xi_2$ (resp., $\xi \equiv \neg (\xi_1
\Or \xi_2)$) is in $\Hf$ (i.e., $\xi$ is a conjunctive formula), then
we put $\Hf := (\Hf \setminus \{\xi\}) \cup \{\xi_1,\xi_2\}$
(resp., $\Hf := (\Hf \setminus \{\xi\}) \cup \{\neg\xi_1,\neg\xi_2\}$),

\item if $\xi \equiv \xi_1 \Or \xi_2$ (resp., $\xi \equiv \neg (\xi_1
\And \xi_2)$) is in $\Hf$ (i.e., $\xi$ is a disjunctive formula), then
we choose a $\xi_i$, $i \in \{1,2\}$, such that $H_{\varphi}\cup
\{\xi_i\}$ (resp., $H_{\varphi}\cup
\{\neg\xi_i\}$)
is satisfiable and put $\Hf := (\Hf \setminus \{\xi\}) \cup
\{\xi_i\}$ (resp., $\Hf := (\Hf \setminus \{\xi\}) \cup
\{\neg\xi_i\}$),

\item if $\xi \equiv \neg(\forall z_1)\ldots (\forall z_n)\varphi_0$
is in $\Hf$, then $\Hf := (\Hf \setminus \{\xi\}) \cup
\{\neg(\varphi_0)_{\bar{z}_1,\ldots,\bar{z}_n}^{z_1,\ldots,z_n}\}$,
where $\bar{z}_1,\ldots,\bar{z}_n$ are newly introduced variables of
sort 0,

\item if $\xi \equiv \neg(\forall Z_1)\ldots (\forall Z_m)\varphi_1$
is in $\Hf$, then $\Hf := (\Hf \setminus \{\xi\}) \cup
\{\neg(\varphi_1)_{\bar{Z}_1,\ldots,\bar{Z}_m}^{Z_1,\ldots,Z_m}\}$,
where $\bar{Z}_1,\ldots,\bar{Z}_m$ are new variables of sort 1.
\end{enumerate}
It is easy to see that the above construction terminates in $\mathcal{O}(|\varphi|)$
steps.
Let us put $\psi \equiv \bigwedge_{\xi \in \Hf}\xi$. Clearly
\begin{enumerate}[~~~~(a)]
\item $\psi$ is a satisfiable $(\TLQSR)^h$-formula,

\item $|\psi| = \mathcal{O}(|\varphi|)$, and

\item $\psi \rightarrow \varphi$ is a valid $\TLQSR$-formula.
\end{enumerate}
In the light of (a)--(c) above, to prove that our problem is in
\textsf{NP}, we only have to construct in nondeterministic polynomial
time a $\TLQSR$-interpretation and show that we can check in polynomial
time that it actually satisfies $\psi$.

Let $\model = (D,M)$ be a $\TLQSR$-model for $\psi$ and let $\model^*
= (D^*,M^*)$ be the relativized interpretation of $\model$ with respect
to a domain $D^{*}$, hence such that
$|D^{*}| = \mathcal{O}(|\psi|^{h+1})$, since $\psi$ is a
$(\TLQSR)^h$-formula (cf.\ Theorem \ref{correctness} and the
construction described in Sections \ref{ssseUniv} and \ref{machinery}).

In the light of the remarks just before Definition~\ref{def:hlang}, to
complete our proof we just have to verify that
\begin{itemize}
\item $M^* A \subseteq \pow_{< h}(D^*)$, for any free variable $A$
of sort 2 in $\psi$ (which entails that $|M^* A| =
\mathcal{O}(|D^{*}|^{h})$), and

\item $\model^* \models \psi$ can be verified in deterministic
polynomial time.
\end{itemize}
The first statement can easily be checked making the following considerations.
By the formula $\xi_{\pi,h}$, we have that $M^* A_{\pi,h} = \pow_{< h}(D^*)$.
Concerning the other $A_i$s of sort 2 occurring in $\psi$, we just have to notice that $\psi$
must contain a conjunct $\psi_i$ associated to $A_i$ that, together with $\xi_{\pi,h}$ ensures that
$M^* A_i \subseteq \pow_{< h}(D^*)$. The proof of the second statement follows from
the fact that each quantified subformula of $\psi$ has the quantifier prefix bounded by $h$ and
that each quantified formula of level 2 has its quantified variables of level 1 ranging in $\pow_{< h}(D^*)$.

Hence the satisfiability problem for $(\TLQSR)^h$-formulae is in \textsf{NP}, and since it is also
\textsf{NP}-hard, it follows that it is
\textsf{NP}-complete.
 \end{proof}
In the next section we will show how the $\TLQSR$ fragment can be used
to formalize the modal logic $\Sc$.
%

\subsubsection{Formalization of $\Sc$ in $\TLQSR$}\label{sec:modalSc}
Let us start with some preliminary notions on modal logics.
The \emph{modal language} $\mathsf{L}_{M}$ is based on a countably
infinite set of propositional letters ${\mathcal P} =
\{p_1,p_2,\ldots\}$, the classical propositional connectives `$\neg$',
`$\And$' , and `$\Or$', the modal operators `$\square$', `$\lozenge$'
(and the parentheses).  $\mathsf{L}_{M}$ is the smallest set
such that ${\mathcal P} \subseteq \mathsf{L}_{M}$, and such that if
$\varphi,\psi \in \mathsf{L}_{M}$, then $\neg \varphi$, $\varphi \And
\psi$, $\varphi \Or \psi$, $\square \varphi$, $\lozenge\varphi\in
\mathsf{L}_{M}$.  Lower case letters like $p$ denote elements of
${\mathcal P}$ and Greek letters like $\varphi$ and $\psi$ represent
formulae of $\mathsf{L}_{M}$.  Given a formula $\varphi$ of
$\mathsf{L}_{M}$, we indicate with $\SubF(\varphi)$ the set of the
subformulae of $\varphi$.

A \emph{normal modal logic} is any subset of $\mathsf{L}_{M}$
which contains all the tautologies and the axiom
$$
\K: \,\, \square(p_1 \rightarrow p_2) \rightarrow (\square p_1 \rightarrow \square p_2) \, ,
$$
and which is closed with respect to modus ponens, substitution, and
necessitation (the reader may consult a text on modal logic like \cite{ModLog01} for more details).

A \emph{Kripke frame} is a pair $\langle W,R\rangle$ such that $W$ is
a nonempty set of possible worlds and $R$ is a binary relation on $W$
called \emph{accessibility relation}.  If $R(w,u)$ holds, we say that
the world $u$ is accessible from the world $w$.  A \emph{Kripke model}
is a triple $\langle W,R,h\rangle$, where $\langle W,R\rangle$ is a
Kripke frame and $h$ is a function mapping propositional letters into
subsets of $W$.  Thus, $h(p)$ is the set of all the worlds where $p$
is true.

Let $\Kr = \langle W,R,h\rangle$ be a Kripke model and let $w$ be a
world in $\Kr$. Then, for every $p \in \mathcal{P}$ and for every
$\varphi,\psi \in {\mathsf{L}_{M}}$, the relation of satisfaction
$\models$ is defined as follows:
\begin{itemize}
\item $\Kr, w \models p$ iff $w \in h(p)$;

\item $\Kr, w \models \varphi \Or \psi$ iff $\Kr, w \models \varphi$ or $\Kr, w \models \psi$;

\item $\Kr, w \models \varphi \And \psi$ iff $\Kr, w \models \varphi$ and $\Kr, w \models \psi$;

\item $\Kr, w \models \neg \varphi$ iff $\Kr, w \not\models \varphi$;

\item $\Kr, w \models \square\varphi$ iff $\Kr, w' \models\varphi$,
for every $w'\in W$ such that $(w,w')\in R$;

\item $\Kr, w \models \lozenge\varphi$ iff there is a $w'\in W$ such that $(w,w')\in R$ and $\Kr, w' \models\varphi$.
\end{itemize}
A formula $\varphi$ is said to be \emph{satisfied} at $w$ in $\Kr$ if
$\Kr, w \models \varphi$; $\varphi$ is said to be \emph{valid} in
$\Kr$ (and we write $\Kr \models \varphi$), if $\Kr,w \models
\varphi$, for every $w \in W$.

The smallest normal modal logic is $\K$, which contains only the modal
axiom $\K$ and whose accessibility relation $R$ can be any binary
relation.  The other normal modal logics admit together with $\K$
other modal axioms drawn from the
ones in Table \ref{tab:modalax}.

%
%
%

\begin{table}[tb]
\begin{center}
\begin{tabular}{|l|l|l|}
  \hline
Axiom & Schema & Condition on $R$ \\\hline
$\T$  & $\square p \rightarrow p$ & Reflexive \\
$\Ac$ & $\lozenge p \rightarrow \square\lozenge p$ & Euclidean \\
$\B$  & $p \rightarrow \square\lozenge p$ & Symmetric \\
$\Aq$ & $\square p \rightarrow \square \square p$ & Transitive \\
$\D$  & $\square p \rightarrow \lozenge p$ & Serial: $(\forall w)(\exists u)R(w,u)$\\
\hline
\end{tabular}
\caption{\label{tab:modalax} Axioms of normal modal logics}
\end{center}
\end{table}

In this paper we analyze the modal logic $\Sc$ which is the strongest
normal modal system.
It can be obtained from the logic $\K$ in several ways.  One of them
consists in adding axioms
$\T$ and $\Ac$ from Table \ref{tab:modalax} to the logic $\K$.  Given a
formula $\varphi$, a Kripke model $\Kr =\langle W,R,h\rangle$, and a world $w \in W$, the
semantics of the modal operators can be defined as follows:
\begin{itemize}
\item $\Kr,w \models \square \varphi$ iff $\Kr,v \models \varphi$, for every $v \in W$,
\item $\Kr,w \models \lozenge \varphi$ iff $\Kr,v \models \varphi$, for some $v \in W$.
\end{itemize}
This makes it possible to translate a formula $\varphi$ of $\Sc$ into the $\TLQSR$ language.

For the purpose of simplifying the definition of the translation
function $\TSc$ given below, the concept of ``empty formula'' is
introduced, to be denoted by $\Lambda$, and not interpreted in any
particular way.  The only requirement on $\Lambda$ needed for the
definition given next is that $\Lambda \And \psi$ and $\psi \And
\Lambda$ are to be considered as syntactic variations of $\psi$, for
any $\TLQSR$-formula $\psi$.

For every propositional letter $p$, let $\TScuno(p) = X_{p}^{1}$,
where $X_{p}^{1}\in {\cal V}_1$, and let $\TScdue : \Sc \rightarrow
\TLQSR$ be the function defined recursively as follows:
\begin{itemize}
\item $\TScdue(p) = \Lambda$,\\

\item $\TScdue(\neg \varphi) = (\forall z)(z \in X_{\neg \varphi}^{1}
\leftrightarrow \neg(z \in X_{\varphi}^{1})) \And \TScdue(\varphi)$,\\

\item $\TScdue(\varphi_1 \And \varphi_2) = (\forall z)(z \in
X_{\varphi_1 \And \varphi_2}^{1} \leftrightarrow (z \in
X_{\varphi_1}^{1} \And z \in X_{\varphi_2}^{1})) \And
\TScdue(\varphi_1) \And \TScdue(\varphi_2)$,\\

\item $\TScdue(\varphi_1 \Or \varphi_2) = (\forall z)(z \in
X_{\varphi_1 \Or \varphi_2}^{1} \leftrightarrow (z \in
X_{\varphi_1}^{1} \Or z \in X_{\varphi_2}^{1})) \And
\TScdue(\varphi_1) \And \TScdue(\varphi_2)$,\\

\item $\TScdue(\square \varphi) = $

\hfill $(\forall z)(z \in X_{\varphi}^{1})
\rightarrow (\forall z)(z \in X_{\square\varphi}^{1})
\And \neg (\forall z)(z \in X_{\varphi}^{1}) \rightarrow
(\forall z)\neg(z \in X_{\square\varphi}^{1}) \And
\TScdue(\varphi)$,\\

\item $\TScdue(\lozenge \varphi) = $

\hfill $\neg (\forall z)\neg (z \in
X_{\varphi}^{1}) \rightarrow (\forall z)(z \in
X_{\lozenge\varphi}^{1})
\And (\forall z) \neg (z \in X_{\varphi}^{1}) \rightarrow
(\forall z)\neg(z \in X_{\lozenge\varphi}^{1}) \And \TScdue(\varphi)$,
\end{itemize}
where $\Lambda$ is the empty formula and $X_{\neg \varphi}^{1}, X_{\varphi}^{1}, X_{\varphi_1 \And \varphi_2}^{1}, X_{\varphi_1 \Or \varphi_2}^{1}, X_{\varphi_1}^{1}, X_{\varphi_2}^{1} \in {\cal V}_1$.

Finally, for every $\varphi$ in $\Sc$, if $\varphi$ is a propositional letter in $\mathcal{P}$ we put
$\TSc(\varphi) = \TScuno(\varphi)$, otherwise $\TSc(\varphi) = \TScdue(\varphi)$.

Even though the accessibility relation $R$ is not used in the translation, we can give its formalization in the $\TLQSR$ fragment
by introducing the collection variable $A_R$ and the following related formulae:
\begin{itemize}
\item $\psi_1 = (\forall Z)(Z \in A_R \rightarrow Z \in A_{\pi,3})$,
\item $\chi_1 = (\forall Z)(Z \in A_{\pi,3}$

$\phantom{\chi_1 = (\forall Z)}\rightarrow (Z \in A_R \leftrightarrow (\forall z_1)(\forall z_2)
  (\forall z_3)((z_1 \in Z \And z_2 \in Z \And z_3 \in Z)$

      $\hfill\rightarrow (z_1 = z_2 \Or z_2 = z_3 \Or z_1 = z_3)).$
\end{itemize}
%
%
%
Clearly $\TSc(\varphi)$ and the formulae above belong to $\TLQSR$ and, in particular, to $(\TLQSR)^3$.
Correctness of the above translation is guaranteed by the following
lemma.
\begin{lemma}\label{leS5}
For every formula $\varphi$ of the logic $\Sc$, $\varphi$ is
satisfiable in a model $\Kr = \langle W,R,h\rangle$ iff there is a
$\TLQSR$-interpretation satisfying $x \in X_{\varphi}$.\eod
\end{lemma}
\begin{proof}
Let $\bar{w}$ be a world in $W$. We construct a $\TLQSR$-interpretation $\model=(W,M)$ as follows:
\begin{itemize}
\item $Mx = \bar{w}$,
\item $M X_{p}^{1} = h(p)$, where $p$ is a propositional letter and $X_{p}^{1} = \TSc(p)$,
\item $M\TSc(\psi) = \true$, for every $\psi\in \SubF(\varphi)$, where $\psi$ is not a propositional letter.
\end{itemize}
To prove the lemma, it would be enough to show that $\Kr,\bar{w}
\models \varphi$ iff $M \models x \in X_{\varphi}^{1}$.  However, it
is more convenient to prove the following more general property:\\

\emph{Given a $w \in W$, if $y \in {\cal V}_0$ is such that $My = w$,
then
$$\Kr,w \models \varphi \mbox{ iff } M \models y \in X_{\varphi}^{1},$$
}
which we do by structural induction on $\varphi$.
\begin{description}
\item \textbf{Base case:} If $\varphi$ is a propositional letter, by
definition, $\Kr,w \models \varphi$ iff $w \in h(\varphi)$.  But this
holds iff $My \in MX_{\varphi}^{1}$, which is equivalent to $M \models
y \in X_{\varphi}^{1}$.

\item \textbf{Inductive step:} We consider only the cases in which
$\varphi = \square \psi$ and $\varphi = \lozenge \psi$, as the other
cases can be dealt with similarly.
\begin{itemize}
\item If $\varphi = \square \psi$, assume first that $\Kr,w \models
\square \psi$.  Then $\Kr,w \models \psi$ and, by inductive
hypothesis, $M \models y \in X_{\psi}^{1}$.  Since $M \models
\TSc(\square \psi)$, it holds that $M \models (\forall z_1)(z_1 \in
X_{\psi}^{1})\rightarrow (\forall z_2)(z_2 \in X_{\square\psi}^{1})$.
Then we have $M[z_1/w,z_2/w]\models (z_1 \in X_{\psi}^{1}) \rightarrow
(z_2 \in X_{\square\psi}^{1})$ and, since $My = w$, we have also that
$M \models (y\in X_{\psi}^{1}) \rightarrow (y \in
X_{\square\psi}^{1})$.  By the inductive hypothesis and by modus
ponens we obtain $M \models y \in X_{\square\psi}^{1}$, as required.

On the other hand, if $\Kr,w \not\models \square \psi$, then $\Kr,w
\not\models \psi$ and, by inductive hypothesis, $M \not\models y \in
X_{\psi}^{1}$.  Since $M \models \TSc(\square \psi)$, then $M \models
\neg(\forall z_1)(z_1 \in X_{\psi}^{1})\rightarrow (\forall
z_2)\neg(z_2 \in X_{\square\psi}^{1})$.  By the inductive hypothesis
and some predicate logic manipulations, we have $M \models \neg(y
\in X_{\psi}^{1})\rightarrow \neg(y \in X_{\square\psi}^{1})$, and by
modus ponens we infer $M\models \neg (y \in X_{\square\psi}^{1})$, as
we wished to prove.

\item Let $\varphi = \lozenge \psi$ and, to begin with, assume that
$\Kr,w \models \lozenge \psi$.  Then, there is a $w'$ such that
$\Kr,w' \models \psi$, and a $y' \in {\cal V}_0$ such that $My' = w'$.
Thus, by inductive hypothesis, we have $M \models y' \in
X_{\psi}^{1}$ and, by predicate logic, $M \models \neg(\forall z_1)\neg(z_1 \in
X_{\psi}^{1})$.  By the very definition of $M$, $M \models
\TSc(\lozenge \psi)$ and thus $M \models \neg(\forall z_1)\neg(z_1 \in
X_{\psi}^{1})\rightarrow (\forall z_2)(z_2 \in X_{\lozenge\psi}^{1})$.
Then, by modus ponens we obtain $M \models (\forall z_2)(z_2 \in
X_{\lozenge\psi}^{1})$ and finally, by
predicate logic, %

$M \models y \in X_{\lozenge\psi}^{1}$.

On the other hand, if $\Kr,w \not\models \lozenge \psi$, then $\Kr,w'
\not\models \psi$, for any $w' \in W$ and, since $w' = My'$ for any
$y' \in {\cal V}_0$, it holds that $M \not\models y' \in X_{\psi}^{1}$
and thus, by predicate logic, $M \models (\forall z_1)\neg(z_1 \in
X_{\psi}^1)$.

Reasoning as above, $M \models (\forall z_1)\neg(z_1 \in
X_{\psi}^{1})\rightarrow (\forall z_2)\neg(z_2 \in
X_{\lozenge\psi}^{1})$ and, by modus ponens, $M \models (\forall
z_2)\neg(z_2 \in X_{\lozenge\psi}^{1})$.  Finally, by predicate logic,
$M \not\models y \in X_{\lozenge\psi}^{1}$, as required.
\end{itemize}
\end{description}
\end{proof}
It can be checked that $\TSc(\varphi)$ is polynomial in the size of
$\varphi$ and that its satisfiability can be verified in
nondeterministic polynomial time since the formula $\xi_{W} \And \xi_{\pi,3}\And \psi_1 \And (\chi_1 \And \TSc(\varphi))$, where
$\xi_W$ denotes $W$, and the other conjuncts are as defined above, belongs to $(\TLQSR)^{3}$.
Consequently, considering that $\Sc$ was proved \textsf{NP}-complete in \cite{Lad77}, the decision algorithm presented in this paper together
with the translation function introduced above can be considered an
optimal procedure (in terms of its computational complexity class) to
decide the satisfiability of any formula $\varphi$ of $\Sc$.
%
%

\section{Conclusions and future work}\label{conclusions}
We have presented the three-sorted stratified set-theoretic fragment
$\TLQSR$ and have given a decision procedure for its satisfiability
problem. Then, we singled out a family of sublanguages of $\TLQSR$, $\{(\TLQSR)^h\}_{h \geq 2}$,
characterized by imposing further constraints in the construction of the formulae, we proved that each language
in the family has the satisfiability problem \textsf{NP}-complete, and we showed that the modal logic $\Sc$
can be formalized in $(\TLQSR)^{3}$.

Techniques to translate modal formulae in set theoretic terms have already
been proposed in \cite{BaMo96}, in the context of hyperset theory, and in \cite{DMP95}
in the ambit of weak set theories not involving the axiom of extensionality and the axiom
of foundation.

We further intend to study the possibility of formulating non-classical logics
in the context of well-founded set theory constructing suitable extensions
of the $\TLQSR$ fragment. In particular, we plan
to introduce in our language a notion of ordered
pair and the operation of composition for binary relations.

%


We also plan to extend the language so as
it can express the set theoretical construct of general union, thus being able to
subsume the theory $3LSSPU$.  Another direction of future investigations concerns
$n$-sorted languages involving also
constructs to express ordered $n$-uples of individuals.



\bibliographystyle{plain}

\begin{thebibliography}{1}

\bibitem{BaMo96} J. Barwise and L. Moss.
\newblock{\em Vicious circles}.
\newblock Vol. 60 of {\em CSLI Lecture notes.} CSLI, Stanford, CA, 1996.


\bibitem{STeP} N. S. Bj{\o}rner, A. Browne, E. S. Chang, M.
Col{\'o}n, A. Kapur, Z. Manna, H. B. Sipma, and T. E. Uribe. STeP:
Deductive-algorithmic verification of reactive and real-time
systems.  In \emph{Proc.  Intl.  Conference on Computer Aided
Verification}, volume 1102 of LNCS, pages 415--418.  Springer, 1996.

\bibitem{ModLog01}
P. Blackburn, M. de Rijke, and Y. Venema.
\newblock{\em Modal Logic}.
\newblock Cambridge University Press, Cambridge Tracts in Theoretical Computer Science, 2001.


\bibitem{Can91}
D.~Cantone.
\newblock Decision procedures for elementary sublanguages of set theory: X. Multilevel syllogistic extended by the singleton and powerset operators.
\newblock {\em Journal of Automated Reasoning}, volume 7, number 2, pages 193--230, Kluwer Academic Publishers,
Hingham, MA, USA, 1991.


\bibitem{CanCut90}
D.~Cantone and V.~Cutello.
\newblock A decidable fragment of the elementary theory of relations and some
  applications.
\newblock In {\em ISSAC '90: Proceedings of the international symposium on
  Symbolic and algebraic computation}, pages 24--29, New York, NY, USA, 1990.
  ACM Press.

\bibitem{CanCut93}
D.~Cantone and V.~Cutello.
\newblock {D}ecision procedures for stratified set-theoretic syllogistics.
\newblock In Manuel Bronstein, editor, {\em Proceedings of the 1993
  International Symposium on Symbolic and Algebraic Computation, ISSAC'93
  (Kiev, Ukraine, July 6-8, 1993)}, pages 105--110, New York, 1993. ACM Press.

\bibitem{CCS90}
D.~Cantone, V.~Cutello, and J.~T. Schwartz.  \newblock Decision
problems for Tarski and Presburger arithmetics extended with sets.
\newblock In {\em CSL '90: Proceedings of the 4th Workshop on Computer
Science Logic}, pages 95--109, London, UK, 1991.  Springer-Verlag.

\bibitem{CanFer1995}
D.~Cantone and A.~Ferro
\newblock Techniques of computable set theory with applications to proof verification.
\newblock {\em Comm. Pure Appl. Math.}, pages 901--945, vol. XLVIII, 1995. Wiley.


\bibitem{CFO89}
D.~Cantone, A.~Ferro, and E.~Omodeo.
\newblock {\em Computable set theory}.
\newblock Clarendon Press, New York, NY, USA, 1989.


\bibitem{COO01}
D.~Cantone, E.~Omodeo, and A.Policriti.
\newblock {\em Set Theory for Computing - From decision procedures to
  declarative programming with sets}.
\newblock Springer-Verlag, Texts and Monographs in Computer Science, 2001.

\bibitem{DMP95}
G. D'Agostino, A. Montanari, and A. Policriti.
\newblock A set-theoretic translation method for polimodal
logics.
\newblock {\em Journal of Automated Reasoning}, 3(15): 317--337, 1995.

\bibitem{FerOm1978}
    A.~Ferro and E.G.~Omodeo.
    \newblock {\em An efficient validity test for formulae in extensional two-level syllogistic}.
    Le Matematiche, 33:130--137, 1978.

\bibitem{FOS80}
   A.~Ferro, E.G.~Omodeo and J.T.~Schwartz.
   \newblock {\em Decision procedures for elementary sublanguages of set theory. I. Multi-level syllogistic and some
   extensions}.
   Communication on Pure and Applied Mathematics, 33, 599--608, 1980.

\bibitem{Lad77}
R. Ladner.
\newblock The computational complexity of provability in systems of modal propositional logic.
\newblock{\em SIAM Journal of Computing,} 6: 467-480, 1977.

\bibitem{CPS06}
  E.~G. Omodeo, D.~Cantone, A.~Policriti, and J.~T. Schwartz.
  \newblock A Computerized Referee.  \newblock In {\em Reasoning,
  Action and Interaction in AI Theories and Systems}, Lecture Notes in
  Artificial Intelligence, vol.  4155, pages 117--139, 2006.

\bibitem{SOC07}
J.~Schwartz, D. Cantone, and E.~Omodeo.
\newblock {\em Computational Logic and Set Theory -- Applying Formalized Logic to Analysis}.
\newblock Springer 2011, I-XVII, 1-416.


\end{thebibliography}

\end{document}